%% file: tccn_drlsa.tex
\documentclass[journal, comsoc]{IEEEtran}

\usepackage[T1]{fontenc}

\ifCLASSINFOpdf
\else
\fi
\usepackage{amsmath}
\usepackage{amsthm}
\usepackage[cmintegrals]{newtxmath}
\usepackage{amssymb}
\usepackage{bm}
\usepackage{bbm}
\usepackage{subfigure}
\usepackage{graphicx}
\usepackage{calc}
\usepackage{epstopdf}

\newtheorem{proposition}{Proposition}

\usepackage{verbatim}
\usepackage{enumitem}
\usepackage{color}
\usepackage{array}
\usepackage{cite}

\usepackage{xcolor}
\usepackage{floatrow}
\usepackage{lipsum}

\usepackage[ruled,norelsize]{algorithm2e}
\usepackage{algorithmic}

\input{MyDef-IRCsplit-Zbased}
\input{MyAcronym-IRCsplit-Zbased}
\graphicspath{{../figures/}}

\def\BibTeX{{\rm B\kern-.05em{\sc i\kern-.025em b}\kern-.08em
    T\kern-.1667em\lower.7ex\hbox{E}\kern-.125emX}}
\begin{document}

\title{Deep Reinforcement Learning Based Spectrum Allocation in Integrated Access and Backhaul Networks}

 \author{Wanlu~Lei,~\IEEEmembership{Student~Member,~IEEE,}
 Yu~Ye,~\IEEEmembership{Student~Member,~IEEE,}
 	Ming~Xiao,~\IEEEmembership{Senior~Member,~IEEE,}
	
 	\thanks{ Wanlu Lei, Yu Ye and Ming Xiao are with the School of Electrical Engineering and Computer Science, Royal Institute of Technology (KTH), Stockholm, Sweden (email: wllei@kth.se, yu9@kth.se, mingx@kth.se).
 	
 	Wanlu Lei is with Ericsson AB, Stockholm, Sweden.}
 }
  
 \maketitle

\begin{abstract}
We develop a framework based on \gls{drl} to solve the spectrum allocation problem in the emerging \gls{iab} architecture with large scale deployment and dynamic environment. The available spectrum is divided into several orthogonal sub-channels, and the \gls{mbs} and all \gls{iab} nodes have the same spectrum resource for allocation, where a \gls{mbs} utilizes those sub-channels for access links of associated \gls{ue} as well as for backhaul links of associated \gls{iab} nodes, and an \gls{iab} node can utilize all for its associated \gls{ue}s. This is one of key features in which 5G differs from traditional settings where the backhaul networks were designed independently from the access networks. With the goal of maximizing the sum log-rate of all \gls{ue} groups, we formulate the spectrum allocation problem into a mix-integer and non-linear programming. However, it is intractable to find an optimal solution especially when the IAB network is large and time-varying. To tackle this problem, we propose to use the latest DRL method by integrating an actor-critic spectrum allocation (ACSA) scheme and \gls{dnn} to achieve real-time spectrum allocation in different scenarios. The proposed methods are evaluated through numerical simulations and show promising results compared with some baseline allocation policies. 

\end{abstract}
\glsresetall

\begin{IEEEkeywords}
integrated access and backhaul, spectrum allocation, deep reinforcement learning 
\end{IEEEkeywords}

\section{Introduction}
The global traffic data rate is estimated to continuously increase with an annual rate of 30\% between 2018 to 2024 due to the exponential increase in the number of mobile broadband subscribers such as smartphones and tablets \cite{DBLP:journals/corr/abs-1906-09298}. An effective solution to handle these relentless demands is to increase the network capacity through network densification and frequency reuse. However, the acquisition cost of site and fiber deployments is one of the biggest challenges in implementing densification at a large scale in the networks. According to \cite{willebrand2001fiber}, the fiber optic link costs about 85\% of total installation and trenching. To handle this challenge, the investigation of self-backhauling, {\color{black}{where same spectrum resources can be shared by both the access and the backhaul transmissions,}} has attracted much attention due to its flexibility and cost efficiency compared with fiber-based wired solutions. 
Recently, the emergence of \gls{iab}, {\color{black}{which is also refereed to as self-backhauling networks in the literature,}} has gained much attention, and it becomes one of the key scenarios in 5G and future \gls{mmwave} networks. Actually, 3GPP Release 15 also involved a new \gls{si} to analyze the feasibility of \gls{iab} deployment for \gls{nr}. Unlike traditional wireless backhauling, \gls{iab} nodes can use the same radio resource for access and backhaul links.
Besides, the dense deployment of \gls{iab} nodes together with the aggressive frequency reuse brings about a direct consequence of challenges in providing reliable backhaul transmission and mitigating severe co-tier and cross-tier interference.

The emerging architecture of \gls{iab} has motivated lots of recent research activities from both academia and industries. References \cite{saha2018bandwidth,saha2018integrated} analyze a two-tier \gls{hetnet} with \gls{iab} where a \gls{mbs} provides \gls{mmwave} backhaul to the \gls{iab} nodes, and concluded that due to the the bottleneck of backhaul links between \gls{mbs} and \gls{iab} nodes, it is difficult to improve user rates in an \gls{iab} settings. In \cite{polese2018end}, it evaluates the end-to-end performance of \gls{iab} architecture through system-level full-stack simulation and shows that \gls{iab} can be a viable solution to effectively relay traffic in very congested networks. The investigation in finding optimal routing and scheduling strategies has also been studied in \cite{polese2018distributed}. As many prior works consider \gls{iab} in \gls{mmwave} scenarios due to the advantage of large bandwidth offered by \gls{mmwave}, none has considered the impact of limited backhaul capacity together with spectrum allocation problem in the sub-6 Ghz setup of \gls{iab}, to the best of our knowledge. 

\gls{iab} settings in sub-6 Ghz frequency actually plays a significant role in many practical scenarios in 5G and beyond networks. Such as in public safety applications, where a network extension is required by adding more \gls{iab} nodes to reach higher coverage, or where a temporary \gls{iab} network needs to be deployed collaboratively with existing \gls{mbs} in certain emergency situations. The spectrum allocation schemes in these scenarios are too divergent from the conventional setup and make it difficult to apply any off-the-shelf methods. Some tailored framework and algorithms for spectrum allocation need to be investigated to adapt to the system dynamics and fulfil real time requirements. 

The spectrum allocation problem has been extensively studied in \cite{chandrasekhar2009spectrum,zhuang2018large,zhou2020blockchain,yang2020heterogeneous} and is usually solved by formulating as an optimization problem. Nevertheless, most of these methods need accurate or complete information about the network, such as \gls{csi} or are achieved with very expensive computational complexity. Besides, network dynamics are seldom addressed and many solutions to the optimization problem are solved in only a snapshot of the network or valid in a specific network architecture. These model-dependent schemes are indeed inappropriate for complex and highly time-varying scenarios. In an established network environment, \gls{bs} employs static spectrum allocation strategy, such as full-reuse or fixed orthogonal allocation methods to ease the system computation and implementation complexities. However, ultra dense \gls{iab} environment makes full-spectrum reuse or other static schemes less efficient due to the severe co-tier interference and cross-tier interference introduced by neighboring \gls{bs}s, as shown in Fig.~\ref{fig:iab}. The rate of \gls{ue} associated with \gls{iab} is determined by the minimum rate of backhaul link and access link, which makes the final rate sensitive to the spectrum allocation strategies. When more \gls{iab} nodes are deployed and more spectrum resource becomes available in the IAB network, the solution space for spectrum allocation increases exponentially. To address this issue, we will exploit latest findings in \gls{rl} and \gls{dnn} to develop a scalable and model-free framework to solve this problem. The framework is expected to have capability to effectively adapt with \gls{iab} network topology changes, large-scale sizes and different real time system requirements. We firstly consider a centralized approach in this work, and leave the distributed approach to future work. 

{\color{black}Many exciting news in Artificial Intelligence(AI) have just happened in recent years\cite{zhang20196g}.} Learning-based method in playing Atari 2600 video games outperforms all previous approaches and surpasses human expert, Alpha Go \cite{silver2016mastering} beats the best human players in the game of Go. Very soon the extended algorithm AlphaGo Zero beats AlphaGo by $100-0$ without supervised learning on human knowledge. The magic behind these happenings is the well-known \gls{rl} algorithm. The development of Q-learning is a big breakout in the early days of \gls{rl}\cite{watkins1992q}. The classical Q-learning is demonstrated to perform well in small-size models but becomes less efficient when the network is scaling up. \gls{rl} combined with the state-of-the-art technique \gls{dnn} addresses the problem and its capability of handling large state spaces to provide a good approximation of Q-value inspires remarkable upsurge of research works in wireless communication problems. The deep Q-learning based power allocation problem has been considered in \cite{ghadimi2017reinforcement} and \cite{zhang2018power}. {\color{black}In \cite{zhang2019deep}, it exploits DRL and proposes an intelligent modulation and coding scheme (MCS) selection algorithm for the primary transmission.} A similar approach has also been proposed in \cite{zhuang2018large} for dynamic spectrum access in wireless network. However, directly applying \gls{dqn} to the spectrum allocation problem is not feasible, because the action space can be very large with increasing of the network size and available spectrum, and it consumes much longer time for convergence. Recent work from DeepMind has introduced an actor-critic method that embraces \gls{dnn} to guide decision making \cite{lillicrap2015continuous} for continuous control. This model can be applied to large discrete action space \cite{dulac2015deep} to tasks having up to one million actions. 

Based on our observations, it can be concluded that \gls{drl} is a promising technique for wireless communication systems in the future. On one hand, the large amount of data from intelligent radio can be used for training and predicting purpose, and hence improving system performance by better decision making (spectrum mapping and allocation for \gls{ue}). On the other hand, it is able to handle highly dynamic time-variant system with different network setups and \gls{ue} demands. Moreover, another appealing advantage in \gls{rl} is the flexibility in reward function design. With proper designed reward at each step, the system performance can be improved according to different desired objective. The interaction between the agent (can be a center located controller or distributed \gls{ue}) and the environment are beneficial in the learning process such that the agent can learn to select appropriate strategies to gain more rewards. 

\textcolor{black}{In what follows, we will study the spectrum allocation problem in the \gls{iab} architecture. To the best of our knowledge, we may be the first to investigate practical \gls{drl}-based spectrum allocation framework in an \gls{iab} setup. Our main contributions are as follows:}
{\color{black}
\begin{itemize}
\item We formulate the spectrum resource allocation problem in an IAB setup into a non-convex mix integer programming, the objective of which tries to maximize a generic network utility function.
\item We develop a novel model-free framework based on \gls{ddqn} and actor-critic techniques for dynamically allocating spectrum resource with system requirements. The proposed architecture is simple for implementation and does not require \gls{csi} related information. The training process can be performed off-line at a centralized controller, 
and the updating is required only when significant changes occur in \gls{iab} setup.
    \item We show that with proposed learning framework, the improvement of the existing policy can be achieved with guarantee, which yields a better sum log-rate performance. We evaluate the learning speed of two algorithms and show that actor-critic based learning applying policy gradient for Q-value maximization is more effective in convergence speed than value-based \gls{ddqn} when action-space is large.
    \item We show that the actor-critic based spectrum allocation scheme outperforms in system sum log-rate compared with the static allocation strategy when the network size increases. 
\end{itemize}
}

The remaining of this paper is organized as follows. We first formulate the spectrum allocation problem in \gls{iab} setup in Section II. To solve the optimization problem, we propose two DRL-based algorithms in Section III. Simulation results are presented in Section IV to demonstrate the effectiveness of the proposed approaches. We conclude the paper in Section V.

\section{System Model}\label{sec:sys_model}
Consider a \gls{dl} transmission two-tier \gls{iab} network, where \gls{mbs} $\mbs$ is located at the center of the network with \gls{iab} nodes deployed uniformly within the coverage area. We denote the set of IAB nodes as 
\textcolor{black}{$\mc B^- = \{\sbs \,\vert\, l =  1, 2,...,\Nsbs\}$}. Each \gls{iab} is equipped with two antennas: the receiving antenna at the \gls{mt} side for the wireless backhaul with \gls{mbs}, and the transmitting antenna at \gls{du} side for access  to serve its associated \gls{ue} groups. \gls{iab} nodes are assumed to be \gls{fd} capable of certain self-cancellation ability. \gls{fd} is an important topic in the \gls{iab} setup. Since our focus in the work is a general spectrum allocation framework, the research for the self-cancellation scheme is not in the scope of this work. Let 
\textcolor{black}{$\mc {B} = \{\mbs\} \cup \mc B^-$} denote the set of all base stations (BS) including the \gls{mbs} and all \gls{iab} nodes, where $\card {\mc{B}} = 1+ \Nsbs$. The total bandwidth in which each \gls{bs} can operate is divided into $\Nsb$ orthogonal sub-channels, denoted by $\mc {M} = \{1,2,...,\Nsb\}$. We assume that each \gls{bs} has a representative \gls{ue} group associated to it where all \gls{ue}s in this group are collocated. 
\textcolor{black}{For notation simplicity, we use similar denotation for UEs as in BSs. The UE group associated with DBS is denoted as $u_0$, while the set of UE groups associated with IAB nodes is denoted as $\mc U^- = \{u_l \vert l = 1,2,...,L\}$, where $u_l$ denotes the UE group that is associated with IAB node $b_l$. The complete UE group set is denoted as $\mc U = \{ u_0 \}\cup \mc U^-$. As often described in \gls{iab} networks, the \gls{iab} nodes are treated as \gls{ue}s to its associated parent, i.e., the \gls{mbs} in our setup. Thus, the first-tier receivers set is denoted as $\mc F_1 = \{u_0\} \cup \mc B^-$, while the second-tier receivers set is $\mc U^-$.}

\begin{figure}[t]
	\centering
	\hspace*{-.1in}\includegraphics[width=0.93 \linewidth]{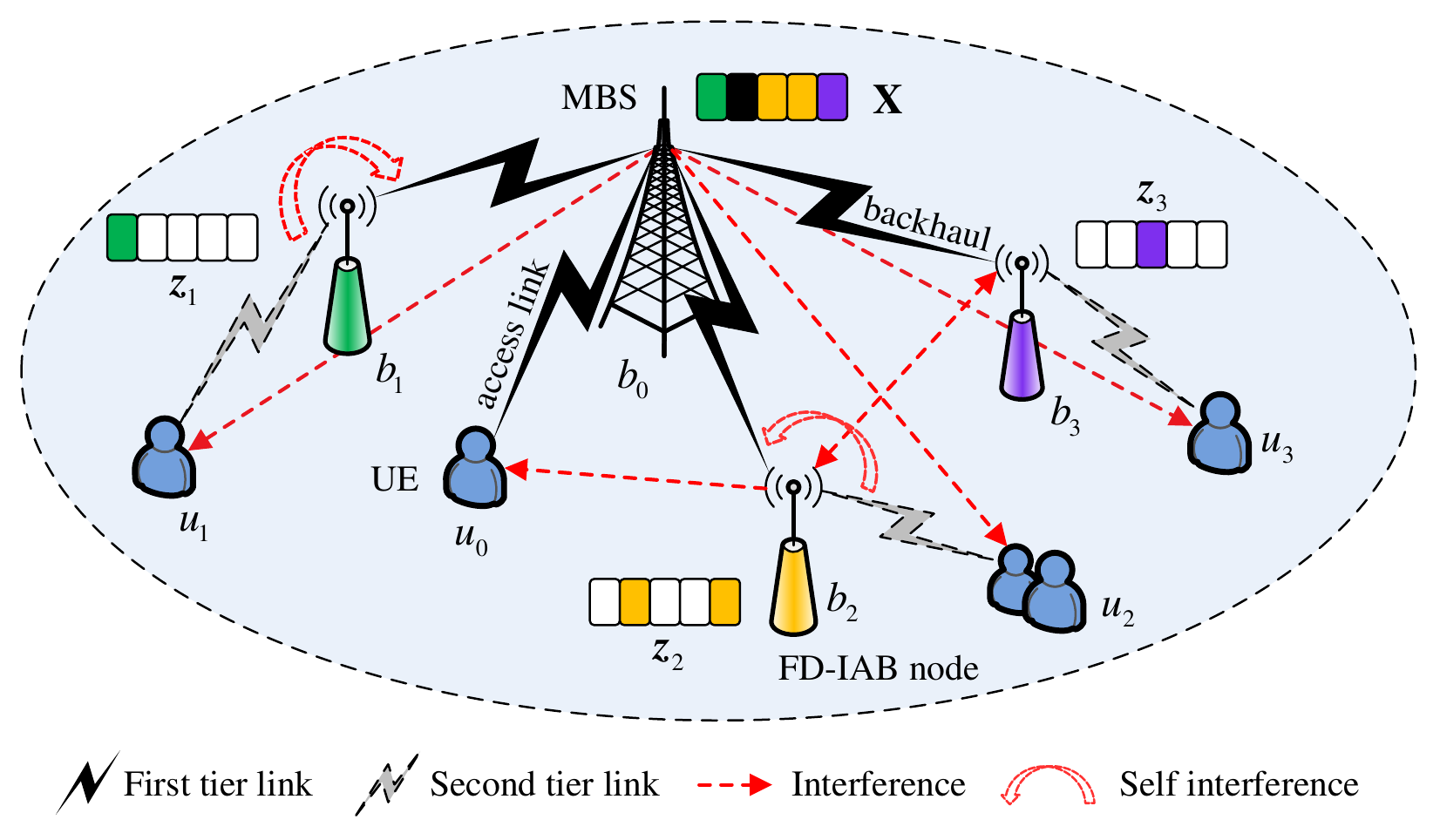}
	\caption{Integrated Access and Backhaul With shared Spectrum Resource.}
	\label{fig:iab}
\end{figure}

Note that each BS has $\Nsb$ spectrum resources for allocation, and we assume that access and backhaul links share the same pool of resource through $\Nsb$ orthogonal sub-channels. The spectrum resource of the \gls{iab} node is dedicated assigned to its associated \gls{ue} group. We denote the spectrum-allocation vector at the $i$-th \gls{iab} as $\sbscvec i = [\sbsc{i}{1},...,\sbsc{i}{M}]^T$, where $\sbsc{i}{m} \in \{0,1\}$, $\Nsbth \in \mc \Nsb$, $i \in \mc B^-$. When the $\Nsbth$-th sub-channel is used by the $i$-th \gls{iab} node, $z_{i}^{m} = 1$ otherwise $\sbsc{i}{m} = 0$. Differently, the spectrum resource at \gls{mbs} can be utilized either by its associated \gls{ue} group for access transmission or the corresponding associated \gls{iab} node for backhaul transmission. We denote the spectrum-allocation vector at \gls{mbs} for its $f$-th receiver as {\color{black}$\mbsvvec{f} = [\mbsv{f}{1},...,\mbsv{f}{M}]^T$}, where $\mbsv{f}{m} \in \{1,0\}$ and $f \in \mc F_1$. For example, if $\mbsv{u_0}{m} = 1$, the \gls{mbs} utilizes the $\Nsbth$-th sub-channel for access link of its associated \gls{ue} $u_0$, and for backhaul link to the $i$-th \gls{iab} node if $x_i^m = 1, i\in \mc B^-$. Thus, for notation simplicity, the allocation mappings is denoted as $\mbsmtx = [\mbsvvec{1}, ..., \mbsvvec{1+L}]^T$ and $\sbsmtx = [\sbscvec{1},...,\sbscvec{\Nsbs}]^T$ for the first-tier receivers and the second-tier receivers, respectively. We assume that \gls{mbs} utilizes all its spectrum resource in the required coverage area. 


Without loss of generality, we assume that several consecutive sub-carriers are grouped into one to form a spectrum sub-channel, and the channel fading is the same within one sub-channel and independent across different sub-channels. At a coherence time period, we denote the downlink channel gain from transmitter $i$ to receiver $j$ at the $\Nsbth$-th sub-channel as 
\begin{equation}
\chgain{i}{j}{m} = \pl{i}{j} \ch{i}{j}{m},
\end{equation}
where $i\in \mc B, j \in  \mc F_1 \cup \mc U^-$. For example, if $i=\mbs$ and $j \in \mc B^-$, $\chgain{i}{j}{m}$ represents the backhaul link channel gain between \gls{mbs} and \gls{iab} node. And $\ch{i}{j}{m}$ is the frequency dependent small-scale fading which undergoes Rayleigh fading, i.e., $\ch{i}{j}{m} \!\sim \text{exp}(1)$. $\pl {i}{j}$ represents the large-scale fading coefficient, and it is a function of distance between $i$ and $j$ including the pathloss and shadowing effect. 

During the transmission session between $i$ to $j$ at the $\Nsbth$-th sub-channel, the interference to the target receiver is introduced by other \gls{bs}s downlink transmissions on the same sub-channel. As \gls{mbs} the transmitter in the first-tier, the received \gls{sinr} of \gls{ue} $u_0$ and \gls{iab} node $b_l$ over the $\Nsbth$-th sub-channel is expressed as 
\begin{equation}
\sinr_{\mbs, \mue}^{m} = \frac{P_{\mbs} \chgain{\mbs}{\mue}{m} \mbsv{\mue}{m}}{\underbrace{ \displaystyle \sum_{i \in \mc B^-} \chgain{i}{\mue}{m} \sbsc{i}{m}}_{\text{co-tier interference}} + \sigma_{u}^2},
\end{equation}
and

\begin{equation}
\sinr_{\mbs, \sbs}^{m} = \frac{P_{\sbs} \chgain{\mbs}{\sbs}{m} \mbsv{\sue}{m}}
{\underbrace{ \displaystyle \sum_{i \in \mc B^- \setminus \sbs} \chgain{i}{\sbs}{m} \sbsc{i}{m}}_{\text{co-tier interference}} + \underbrace {P_{\sbs}\xi_{b_l}  \sbsc{b_l}{m}}_{\text{self-interference}} + \sigma_{\sbs}^2},
\end{equation}
where 
$\xi_{b_l}$ is the self-interference cancellation ability of \gls{iab} node $b_l$, and $\sigma_u^2$ and $\sigma_{\sbs}^2$ are noise power for the \gls{ue} and \gls{iab} node, respectively. 
The received \gls{sinr} for the second-tier \gls{ue} groups is
\begin{equation}
\sinr_{\sbs, \sue}^{m} = \frac{P_{\sbs} \chgain{\sbs}{\sue}{m} \sbsc{\Nsbsth}{\Nsbth}}
{\underbrace{\displaystyle \sum_{i \in \mc B^- \setminus \sbs} \chgain{i}{\sue}{m} \sbsc{i}{m} }_{\text{co-tier interference}} + \underbrace {\chgain{\mbs}{\sue}{m}  \displaystyle \sum_{j \in \mc F_1} \mbsv{j}{m} }_{\text{cross-tier interference}} + \sigma_{u}^2}.
\end{equation}

The instantaneous rate of $\mue$ is determined by the spectrum allocation matrix $\mbsmtx$ and $\sbsmtx$ that is expressed as
\begin{equation}
\rate{\mue}(\mbsmtx, \sbsmtx) = \displaystyle \sum_{m=1}^{\Nsb} \log (1+ \sinr_{\mbs, \mue}^{m} (\mbsmtx, \sbsmtx) ). \label{eq:rate_u0}
 \end{equation}
 Since the \gls{iab} node receives signal from \gls{mbs} and transmits to its associated \gls{ue} group, the instantaneous rate of the second-tier \gls{ue} $\sue$ is decided by the minimum value of the backhaul rate and the access rate, which is given by
\begin{equation}
\rate{\sue}(\mbsmtx, \sbsmtx) = \min(\rate{\mbs,\sbs}(\mbsmtx, \sbsmtx), \rate{\sbs, \sue}(\mbsmtx, \sbsmtx)), \label{eq:rate_ul}
 \end{equation}
where $\rate{\mbs,\sbs}$ represents the backhaul link rate from \gls{mbs} to \gls{iab} node $b_l$, and $\rate{\sbs,\sue}$ is the access link rate for corresponding associated \gls{ue} group of the IAB node $b_l$.


Based on the description above, we are interested in maximizing a generic network utility function $f(x)$ as expressed below

  \begin{subequations}
	\begin{align}
(\text{P}1):~\max_{\bm{X},\bm{Z}}~&\sum_{j \in \mathcal{U}}f(C_{j}(\bm{X},\bm{Z})),\label{eq:max_lp}\\
s.t. ~&C_{j}\geq \Omega_{j},~\forall j\in\mathcal{U};\label{eq:cons_capcity}\\
&\sum_{f\in\mathcal{F}_1}x_f^m=1,~m\in\mathcal{M};\label{eq:cons_v1}\\
&\sum_{f\in\mathcal{F}_1}\sum_{m\in\mathcal{M}}x_f^m \leq M;\label{eq:cons_v2}\\
&\sum_{m\in\mathcal{M}}z_i^m \leq M,~\forall{i}\in\mathcal{B}^-;\label{eq:cons_q1}\\
&x_f^m,z_i^m\in\{0,1 \},~\forall f\in\mathcal{F}_1,\forall i\in\mathcal{B}^-,\forall m\in\mathcal{M},\label{eq:cons_q2} 
	\end{align}	
\end{subequations}
where \eqref{eq:cons_capcity} is Quality of Service (QoS) requirement of each \gls{ue} group in the system, \eqref{eq:cons_v1} and \eqref{eq:cons_v2} are constraints for allocation vector of \gls{mbs}, such that each sub-channel can be only allocated to either access or backhaul transmission. \eqref{eq:cons_q1} and \eqref{eq:cons_q2} are constraints for the \gls{iab} node allocation vector that only $\Nsb$ sub-channels are available to employ. The spectrum allocation problem (P1) is non-convex mix integer programming, which has been shown to be NP-hard\cite{luo2008dynamic}. The solution of (P1) becomes intractable especially when the network size and available sub-channels increase. 

The function $f(x)$ can be designed to obtain a certain network utility. Since we assume that each \gls{bs} always has a \gls{ue} group associated to it. Since each \gls{ue} group requires to reach a certain rate level for transmission purpose at all time step, we consider the objective which aims at maximizing the global system rate. This is known as proportional fairness\cite{srikant2013communication}, and it yields $f(x) = \log(x)$.

To this end, the spectrum resource allocation problem of the \gls{iab} network is to design a spectrum allocation scheme for backhaul and access link transmission using the same resource pool. Spectrum resources expressed through binary variables $\mbsv{f}{m}$ and $\sbsc{i}{m}$ ($\forall f \in \mc F_1, \forall i \in \mc B^-$, $\forall\Nsbth \in \mc \Nsb$) are designed to maximize the network utility function as described in (P1).

\section{Reinforcement Learning in Resource Allocation}\label{sec:rru_reduction}
\subsection{Background of Reinforcement Learning}
The essential idea behind \gls{rl} is to learn good policy through trail-and-error interaction with a dynamic environment \cite{sutton1998introduction,kaelbling1996reinforcement}. \gls{rl} dates back to the early days of work in psychology, neuroscience and computer science. It receives massive attention recently in the machine learning communities. It has been widely used in decision making problems, mainly for the reason that it does not require prior knowledge of the system model and can be effectively adapted to stochastic transitions in the system. The \gls{rl} algorithm is initially designed to find the optimal policy with a fully observable Markovian environment. RL methods can also be applied in various settings such as non-Markovian environment and and is shown to have good performance.

Let $\mc S$ denote a set of possible states in the \gls{iab} network environment, and $\mc A$ denote a discrete set of actions. The problem can be modeled as an \gls{mdp}. The transition probability distribution is the set $\mc P: \mc S \times \mc A \times \mc S \to \mathbb R$, and the reward function is $r: \mc S \times \mc A \to \mathbb R$. Assuming discrete time steps, at each coherence time $t$, the central controller of the \gls{iab} network observes a state $\state {t} \in \mc S$, and takes an action $\actionv {t} \in \mc A$ according to a certain policy $\pi: \mc S \to \mc A$. The policy $\pi(\actionv{}\vert \state{})$ is the probability of taking action $\actionv{}$ conditioned on the current state $\state{}$. Therefore, the policy function must satisfy $\sum_{\action{}\in \mc A} \pi(\actionv{}\vert\state{}) = 1$. The controller, as the agent, receives an immediate reward $\reward {}$ and the environment evolves to next state $\nextstate$ with probability $P(\nextstate{} \vert \state{},\actionv{})$. Here let $R_t$ be the future cumulative discounted reward at time $t$ expressed as
\begin{equation}
   R_t = \displaystyle \sum_{\tau=0}^{\infty} \gamma^{\tau} \reward{t+\tau+1},
  \end{equation}
where $\gamma \in [0,1]$ is a discounted factor for future reward. The agent is only concerned with an immediate reward if $\gamma$ is set to $0$, which is considered as myopic. As $\gamma$ approaches $1$, the agent takes future reward into account more strongly. The objective of \gls{rl} agent is to find a policy $\pi$ that maximizes the expected accumulative discounted reward given the state-action pair $(\state{},\actionv{})$ at time $t$. This is also defined as Q-function in the classical Q-learning algorithm
\begin{equation}
	Q^{\pi}(\state{},\actionv{}) = \mathbb E_{\pi} \ls R_t\vert \state{t} =\state{}, \actionv{t} = \actionv{} \rs. \label{eq:Qfunc_expect}
\end{equation}

A commonly used off-line algorithm to find the optimal stochastic policy function $\pi^{*}(\actionv{}\vert\state{})$ takes greedy search, which is expressed as
\begin{equation}
	\pi^{*}(\actionv{}\vert\state{}) = \left\{
	\begin{aligned}
	1, & ~\text{if} ~ \action{} = 	 \arg \max_{a\in \mathcal{A}} ~ Q^{\pi}(\state{},\actionv{}); \\
	0, &~ \text {o.w}.
	\end{aligned}\right.\label{eq:policy_func}
\end{equation}
The policy function is written in $\pi(\state{})$ for deterministic transitions. Further analysis on policy improvement is discussed in Sec.~\ref{subsec:policy_improv}. It has been proved in \cite{tsitsiklis1994asynchronous} that Q-learning can effectively converge to the optimal policy with some iterative algorithms to achieve the optimal Q-value as
\begin{equation}
	Q^*(\state{},\actionv{}) =  r(\state{},\actionv{}) + \gamma \displaystyle \sum_{\state{} \in \mc S} P(\nextstate \vert \state{},\actionv{}) \, \underset {\nextaction}{\text{max}}~Q^*(\nextstate,\nextaction). \label{eq:Q_optim}
\end{equation} 

The classical Q-learning constructs a lookup table for all state-action pairs and updates the table at each step to approach optimal Q-value. The updating of Q-value for each state-action pair usually applies the $\epsilon$-greedy policy. With probability $\epsilon$ the agent picks a random action, and with probability $1-\epsilon$ the agent takes the action that gives the maximum Q-value at the current state. This is also referred as \emph{exploration} and \emph{exploitation} in \gls{rl}. Normally, at the beginning of the iteration, $\epsilon$ is set to $1$ to enable the agent to have more chances to explore the environment with different new actions, and then it decays with the iteration episodes. After obtaining new experience through interacting with the environment by taking action $\actionv{}$ at state $\state{}$, the agent transits to next state $\nextstate$ and updates the corresponding new Q-value 
\begin{align}
	Q(\state{},\actionv{})  \gets &  Q(\state{},\actionv{}) +  \nonumber \\
	& \beta \underbrace {\bigl( \reward{}(\state{},\actionv{}) + \gamma \underset {\nextaction}{\max}~Q(\nextstate,\nextaction) - Q(\state{},\actionv{}) \bigr)}_{\text{temporal difference error}}, \label{eq:Q_table}
\end{align}
where $\beta$ is the learning rate. And $\reward{}(\state{},\actionv{}) + \gamma   {\max}_{\nextaction}~Q(\nextstate,\nextaction)$ is the learned value for each step, the former of which is the immediate reward obtained from the environment after taking the action $\actionv{}$, and the latter is the estimated discounted maximum Q-value after transiting to the state $\nextstate$ by taking action $\actionv{}'$. The updating rule is called \gls{td} method since its changes are based on the estimated error between two different times. 

The above value-based Q-learning method performs well when dealing with small state and action space. However, it becomes computationally infeasible when the state space and action space are large, such as in the spectrum allocation problem at hand. This is due to that the storage of the lookup table in \eqref{eq:Q_table} for all state-action pairs is intolerable for large-scale problems. Besides, many states are rarely visited as the state space increases, which consequently leads to much longer time to converge to the optimal Q-value. To solve these two problems, deep Q-learning is proposed to improve the classical Q-learning method by applying the \gls{dnn} to approximate Q-function in lieu of the lookup table, which is called \gls{dqn}.

\subsection{DQN and Spectrum Allocation}\label{subsec:dqn}
In \gls{dqn}, the Q-function in \eqref{eq:Qfunc_expect} is approximated by a \gls{dnn} with a weight parameter $\thetaq$ as
\begin{equation}
    \tilde Q(\state{},\actionv{}; \thetaq ) \approx Q(\state{},\actionv{}).
\end{equation}
The essence of \gls{dqn} is to determine the weights $\thetaq$ for Q-function. Similar to Q-learning, at each step, the agent collects experience through interacting with the \gls{iab} network and stores this experience as the form $(\state{i},\actionv{i},\reward{i}, \nextstate_{i})$ in a data set $D$. Then \gls{dqn} updates its $\thetaq$ to minimize the loss {\color{black}by sampling a minibatch of size $N$ from data set $D$ as}
\begin{equation}
	\thetaq =  \arg\min_{\thetaq} ~ L(\thetaq) = \arg\min_{\thetaq} ~\displaystyle \sum_i^N \left(y_i-\tilde Q(\state{i},\actionv{i};\thetaq)\right)^2, \label{eq:loss_naturedqn}
\end{equation}
where 
\begin{equation} \label{eq13}
 y_i = \reward{i} +  \gamma \, \underset {\nextaction_i}{\max}~\tilde Q(\nextstate_i,\nextaction_i; \VEC{\theta}).
 \end{equation}
As \gls{dqn} aims to greatly improve and stabilize the training procedure of Q-learning, it applies two innovative mechanisms: \emph{experience replay} and \emph{periodically update target}. By applying experience replay mechanism, the \gls{dqn} randomly samples previous transitions in the data set $D$ to alleviate the problem of correlated data and non-stationary distributions \cite{DBLP:journals/corr/MnihKSGAWR13}. However, the operator $ {\max}_{\nextaction}~Q(\nextstate,\nextaction; \thetaq)$ in \eqref{eq13} uses the same values as that in evaluating an action, this may lead to overoptimistic value estimates. Therefore, as the "quasi-static target network" method implies in \cite{DBLP:journals/corr/HasseltGS15}, \gls{ddqn} can solve this issue by introducing two \gls{dqn}s: the target \gls{dqn} with weights $\thetatarget$ and the train \gls{dqn} with weights $\thetaeval$ in order to remove upward bias. Parameters $\thetaq^-$ are basically clones from the train network but only updated periodically according to the settings in hyper-parameter. The difference of loss functions between \eqref{eq:loss_naturedqn} and \gls{ddqn} is that the train Q-network is used to select actions while the target Q-network is used to evaluate actions. The target Q-value is calculated as
\begin{equation}
y_{\text{DDQN}} = \reward{}(\state{},\actionv{}) +  \gamma \, \tilde Q\bigl(\nextstate,\underbrace{  \arg\max_{\nextaction\in \mc A} ~\tilde Q(\nextstate,\nextaction; \thetaeval )}_{\text{select action} \, \nextaction}; \thetatarget
\bigr). \label{eq:y_update}
	\end{equation}
Finally, the stochastic gradient descent, which minimizes the loss in \eqref{eq:loss_naturedqn} by applying the target value in \eqref{eq:y_update} is used for training over the mini-batch data set $D$ at a central controller. 
\begin{algorithm}[t]
\caption{\strut Double \gls{dqn} training algorithm for Spectrum Allocation}\label{alg:ddqn}
\begin{algorithmic}[1]
\STATE{Start environment simulator for \gls{iab} network;}

\STATE{Initialize Memory data set $D$;}
\STATE{Initialize the target and the train \gls{dqn} $\thetatarget = \thetaq$;}
\FOR{episode $=1,2,...$}
\STATE{Receive the initial observed state $\state{1}$;}

\FOR{$t = 1:\Nsteps$}
\STATE{Select an action using $\epsilon$-greedy method from the train network $\tilde Q(\state{},\action{}; \thetaeval)$;}
\STATE{Obtains a reward according to Eq.~\eqref{eq:reward}, revolves to new state $\state{t+1}$;}
\STATE{Store transition pair $(\state{t},\action{t},\reward{t},\state{t+1})$ in data set $D$;}
\IF{training is TRUE}
{\color{black}\FOR{$i = 1:N$}
\STATE{Sample a random mini-batch of transitions $(\state{i},\actionv{i},\reward{i},\state{i+1})$ from data set $D$};
\STATE{Set $\{y_i\}$ according to \eqref{eq13}}; 
\ENDFOR
}
\STATE Update train network according to \eqref{eq:loss_naturedqn};
\STATE{Update target network weight $\thetatarget = \thetaeval$};
\ENDIF
\ENDFOR
\ENDFOR
\end{algorithmic}
\end{algorithm}

At each coherence time, the agent builds its state using collected information from its associated \gls{ue} and \gls{iab} nodes. To better control the transmission overhead, we restrict the state information to the \gls{qos} status of all \gls{ue}s. The center unit first computes the rate level for each \gls{ue} according to \eqref{eq:rate_u0} and \eqref{eq:rate_ul} at $t$-th step, and then obtains the state denoted as $\state{t} =\{s_{t,u_0}, s_{t,u_1}, ... , s_{t,u_{1+L}}\}$, where $s_{t,j} \in \{0,1\}$. And $s_{t,j}= 1$ indicates that the demand of the \gls{ue} group $(j\in\mathcal{U})$ has been satisfied with {\color{black}$\rate{t,j} \geq \Omega_{t,j}$} at time $t$, and $s_{t,j} = 0$ otherwise. 
In the simulation, we observe that the agent can learn accurately by using these \gls{qos} information without extra \gls{csi} feedback from \gls{iab} nodes or \gls{ue}s. 
With the goal of optimizing the spectrum allocation scheme to maximize the sum log-rate, we define the action as the corresponding allocation matrix for \gls{mbs} and \gls{iab} nodes as $\actionv{t} = \ls \mbsmtx_t,~ \sbsmtx_t \rs^T$. The action decision for channel allocation at \gls{iab} nodes is to evaluate the contribution of specific channels to the system objective. {\color{black}{For example, if the condition of channel $m$ is weak at node $f\in \mathcal{B}^-$, and contributes more to interfering neighbors than to its associated \gls{ue} signal strength, then this channel is very likely to be set to off at next state, i.e., $x_f^m=0$.}} The reward function is designed to optimize the network objective in (P1) for proportional fairness allocation. For all $0\leq t \leq \Nsteps$, we define the reward as
\begin{equation}
    \reward{t} =  \displaystyle \sum_{j\in \mc U} f \lp \rate{t,j}(\mbsmtx_t , \sbsmtx_t) \rp. \label{eq:reward}
\end{equation}
The \gls{dqn} is trained at the central controller in an off-line manner. The training procedure is stated in Algorithm~\ref{alg:ddqn}. Both target network and train network are initialized in the same manner. We use a memory replay buffer for storing transition samples and then fetch a mini-batch of transition samples from $D$ to train the deep Q-learning networks. Though there are many investigations regarding the sampling method to improve the efficiency and accuracy of the training process \cite{schaul2015prioritized,wang2015dueling}, we apply a random sampling method to take into account all possible experiences while interacting with the \gls{iab} network environment. Through step 14 in Algorithm 1, the weights $\theta^-$ for the target network get periodically updated. 

Although value-based \gls{dqn} learning has been demonstrated very powerful in many application areas, there are several limitations which have been addressed in \cite{sutton2000policy,lillicrap2015continuous}. For instance, as action space becomes large, the evaluation over $\vert \mc A \vert$ for action selection becomes intractable.
As the spectrum allocation problem stated in (P1), the action space of the \gls{mbs} is $(1+\Nsbs)^{\Nsb}$, and  each \gls{iab} has $2 ^ {M}$ actions. It is obvious that the total action space increases exponentially with the number of \gls{iab} nodes and available spectrum. The required training episodes also increase significantly. Hence to improve learning efficiency, we propose a actor-critic learning method that can be more effectively adapting to the size of \gls{iab} networks.

\begin{algorithm}[ht]
\caption{\strut Actor-Critic Spectrum Allocation Algorithm(ACSA)}\label{alg:acsa}
\begin{algorithmic}[1]
\STATE{Start environment simulator for \gls{iab} network;}

\STATE{Initialize Memory data set $D$;}
\STATE{Initialize target critic and actor network, $\thetatarget = \thetaq, \thetapolicy^- = \thetapolicy$;}
\FOR{episode $=1,2,...$ }
\STATE{Observed the initial state $\state{1}$;}
\FOR{$t = 1:\Nsteps$}
\STATE{Select action $\actionv{t}$ from actor network;}
\STATE{Execute action $\actionv{t}$ to the spectrum allocation matrix $\mbsmtx$ and $\sbsmtx$ and obtains reward $\reward{t}$;}
\STATE{Store transition pair $(\state{t},\action{t},\reward{t},\state{t+1})$ in data set $D$;}
\IF{training is TRUE}
{\color{black}\FOR{$i = 1:N$}
\STATE{Sample mini-batch from data set $D$;}
\STATE{Set $y_i = \reward{i} + \gamma \tilde Q(\nextstate_i,\tilde \pi(\nextstate_i;\thetapolicy^-);\thetaq^-)$;}
\STATE{Compute \gls{td} error $\delta_i = y_i - \tilde{Q}(\state{i},\action{i};\thetaq)$}
\ENDFOR
}
\STATE{Update $\thetaq$ for critic network according to \eqref{eq:loss_naturedqn}};
\STATE{Compute sampled policy gradient 
\begin{align}
    \nabla_{\omega}J\approx& \frac{1}{N}\sum_i \nabla_{\actionv{}}\tilde Q(\state{},\actionv{};\thetaq)\vert_{\state{}=\state{i},\actionv{}=\tilde \pi(\state{i};\thetapolicy)} \notag\\&\times \nabla_{\thetapolicy}\tilde \pi(\state{};\thetapolicy)\vert_{\state{}= \state{i}};\notag
\end{align}}
\STATE{Update $\omega$ for actor network by
\begin{align}
    \thetapolicy =& \thetapolicy +  \beta \nabla_{\thetapolicy} J;\notag
\end{align}}
\ENDIF
\STATE{Update the weights of the target networks 
\begin{align}
    \thetaq^- =& \tau \thetaq + (1-\tau)\thetaq^-;\notag\\
    \thetapolicy^- =& \tau \thetapolicy + (1-\tau)\thetapolicy^-;\notag
\end{align}}
\ENDFOR
\ENDFOR
\end{algorithmic}
\end{algorithm}

\begin{figure}[b]
	\centering
	\includegraphics[width=0.9 \linewidth]{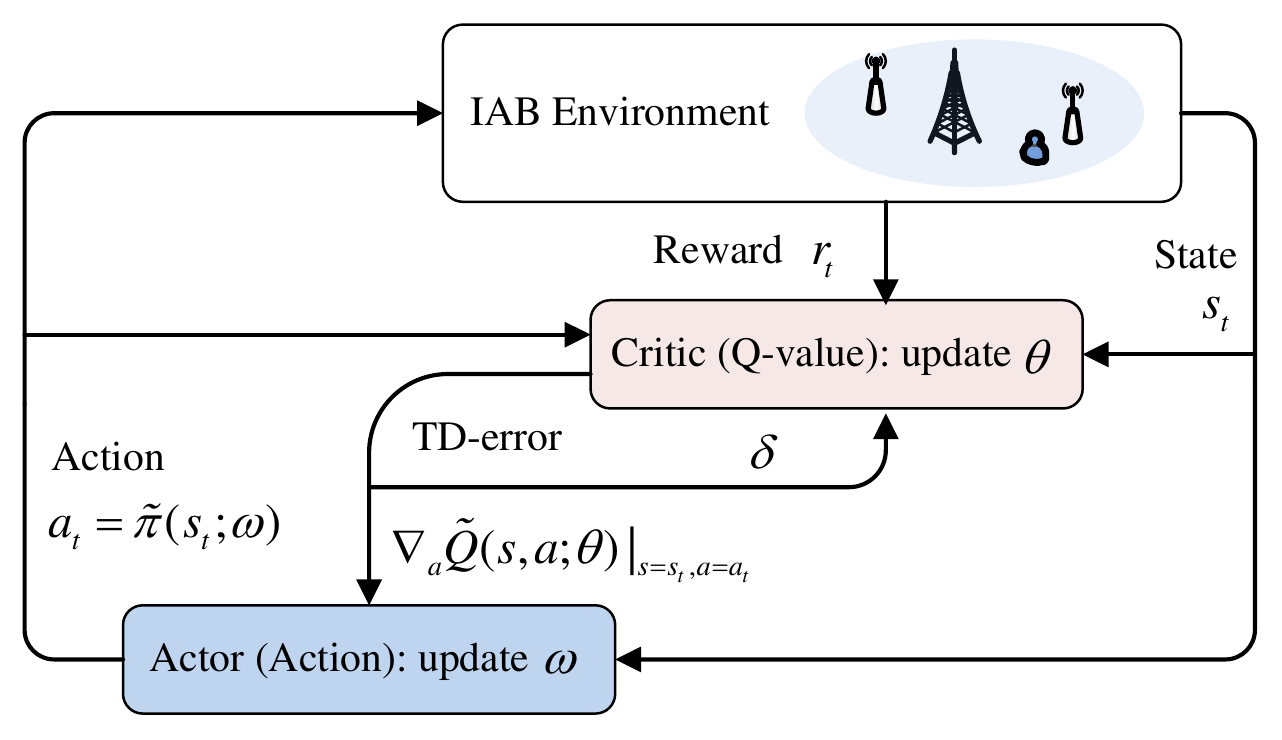}
	\caption{Actor-Critic Spectrum allocation architecture}
	\label{fig:ac_fig}
\end{figure}

\subsection{Actor-Critic Spectrum Allocation Method}
In this section we describe an Actor-Critic Spectrum Allocation (ACSA)-based framework for solving spectrum allocation problem (P1). Unlike the value-based \gls{dqn}, the ACSA model is premised on having two interacting modulus. As shown in Fig.~\ref{fig:ac_fig}, it has two aptly named components: an actor and a critic. The role of the actor is to take current environment state and output an action upon certain policy, which is essentially as normal \gls{dqn} does. The critic is responsible for evaluating how good the action is for the state through returning a score from \gls{dqn} by taking the action and current state as inputs. The critic is updated using \gls{td} errors similar as Q-network. The parameters of actor {\color{black}are} updated by the returning score to improve its performance towards actions that bring higher score in the future. In most actor-critic models, the policy function is modeled as a probability distribution over actions $\mc A$ given the current state and thus it is stochastic. This model has been widely used given support by the policy gradient theorem \cite{sutton2000policy,degris2012model}. However, the appealing simple form for deterministic policy gradient has been demonstrated to be more efficient than the usual stochastic policy gradient in \cite{silver2014deterministic}. Therefore, we model the policy of the spectrum allocation problem as a deterministic decision.
Following DDQN, the Q-function in \eqref{eq:Qfunc_expect} is approximated via \gls{dqn} with a weight vector $\thetaq$, while the actor network uses a similar approach which approximates the policy function with the parameter $\thetapolicy$ as $\tilde \pi(\state{};\thetapolicy)$. 
The goal of the actor network is to determine parameter $\thetapolicy$ in policy function $\tilde \pi(\state{};\thetapolicy)$, which may result in the largest increase of Q-value. According to \cite{silver2014deterministic}, the actor network can be updated by applying the chain rule to the expected cumulative reward $J$ with respect to the actor parameters $\thetapolicy$,
\begin{align}
        \nabla_{\thetapolicy} J  & \approx \mathbb E \ls  \nabla_{\thetapolicy}\tilde Q(\state{},\actionv{};\thetaq) \vert_{\state{}=\state{t},\actionv{}=\tilde \pi(\state{t},\thetapolicy)}\rs \nonumber \\
  & =  \mathbb E \ls \nabla_{\actionv{}}\tilde Q(\state{},\actionv{};\thetaq)\vert_{\state{}=\state{t},\actionv{}=\tilde \pi(\state{t};\thetapolicy)}  \nabla_{\thetapolicy}\tilde \pi(\state{};\thetapolicy)\vert_{\state{}= \state{t}}\rs. \label{eq:gradient}
\end{align}

We formally present the proposed actor-critic spectrum allocation method in Algorithm~\ref{alg:acsa}. Recall that \gls{ddqn} stabilizes the learning of Q-function by experience replay and the {\color{black}{frozen target network, i.e., the parameters of evaluating Q-network are copied to the target Q-network periodically.}} We similarly employ target networks for critic and actor with parameters $\thetaq^-$ and $\thetapolicy^-$ respectively. The target networks are clones of the train actor and critic networks, which are initialized in the same way as the train networks but updated with rate $\tau$ (step 18). In iteration $t$, the agent applies $\epsilon$-greedy exploration method in the actor network to obtain an action. Then the state-action pair is sent into critic train network for updating $\thetaq$.
After receiving the gradient with respect to action output by the critic, the train network for actor updates $\omega$.

\subsection{Improvement of baseline policy} \label{subsec:policy_improv}
In this subsection, we will investigate the improvement of the policy determined by ACSA, which may consequently achieve a better sum log-rate compared with baseline spectrum allocation strategy. 
Consider an \gls{iab} network that adopts a baseline policy $\policybase(\state{})$ which generates the state sequence  $\{\state{t}, \state{t+1},...\state{\Nsteps}\}$ from time $t$. 
We define the expected return under this policy $\policybase$ starting {\color{black}{from $\state{t} = \VEC s$}} as a value function 
\begin{equation}
    \mc V(\state{},\policybase) = \mathbb E\ls R_t \vert \state{t} = \state{} \rs.\label{eq:value_func}
\end{equation}
To show whether the policy $\policynew(\state{};\thetapolicy)$ obtained by ACSA is better than $\policybase$, we consider a special case where an action $\actionv{}^\dagger= \policynew(\state{};\thetapolicy) \neq \policybase(\state{})$ is selected at state $\state{}$. And then the agent follows the baseline policy. We denote $Q^{\policynew}(\state{},\policynew(\state{};\thetapolicy))$ as the value obtained with this behavior. If ACSA is a better strategy, it is natural to obtain $Q^{\policynew}(\state{},\policynew(\state{};\thetapolicy)) \geq \mc V(\state{},\policybase)$.
\begin{proposition}
If the policy $\policynew(\state{};\thetapolicy)$ of ACSA is no worse than an existing policy $\policybase$, it must obtain $\mc V(\state{},\policybase)\leq \mc V(\state{},\policynew)$.
\end{proposition}

\begin{proof}
According to \eqref{eq:value_func} we derive
\begin{align}
  \mc V(\state{},\policybase)  & = \mathbb E\ls\displaystyle \sum_{\tau=0}^{\infty} \gamma^{\tau} \reward{}(\state{},\policybase(\state{})) \vert \state{t} = \state{} \rs\notag \\
    & = \mathbb E\ls \reward{}(\state{},\policybase(\state{}) + \displaystyle \sum_{\tau=0}^{\infty} \gamma^{\tau} \reward{}(\state{t+1},\policybase(\state{t+1}) \vert \state{t} = \state{} \rs\notag \\
    & = \mathbb E \ls \reward{}(\state{},\policybase(\state{}))\vert \state{t} = \state{} \rs +  \gamma \mathbb E \ls R_{t+1} \vert \state{t+1} = \nextstate \rs \notag\\
    & =  \mathbb E \ls \reward{}(\state{},\policybase(\state{}))\vert \state{t} = \state{} \rs + \gamma \mc V(\nextstate{}, \policybase).
\end{align}
Letting $\actionv{}^\dagger= \policynew(\state{};\thetapolicy)$ be the action taken at state $\state{t}$, we have 
	\color{black}
    \begin{align}
   	\mc V(\state{},\policybase)  \leq & Q^{\policybase}\left(\state{},\policynew(\state{}; \thetapolicy)\right) \notag\\
     = &\mathbb E \ls R_t\vert \state{t} =\state{}, \actionv{t} = \policynew(\state{};\thetapolicy)\rs\notag \\
	= & \mathbb E_{\bar{\pi}} \ls \reward{}(\state{},\pi^{\dagger}(\state{}; \omega))) + \gamma \mc V(\state{t+1},\bar{\pi}(\state{t+1})) \vert \state{t} =\state{}\rs\notag \\
	\leq & \mathbb E_{\bar {\pi}} \ls \reward{}(\state{},\pi^{\dagger}(\state{};\omega)) + \gamma Q^{\bar{\pi}}(\state{t+1},\pi^{\dagger}(\state{t+1})) \vert \state{t} =\state{}\rs\notag \\
	= & \mathbb E_{\bar {\pi}} [ \reward{}(\state{},\pi^{\dagger}(\state{};\omega))  + \gamma \mathbb E [ r(\VEC s_{t+1}, \pi ^{\dagger}) \notag \\
	& + \mc V(\VEC{s}_{t+2}, \bar {\pi}) \vert \VEC s_{t+1}, \VEC a_{t+1} = \pi^{\dagger}(\state{};\omega) ] \vert \state{t} = \state{}] \notag\\
	= & \mathbb E_{\bar {\pi}} [ \reward{}(\state{},\pi^{\dagger}(\state{};\omega)) +  \gamma r(\VEC s_{t+1}, \pi ^{\dagger})\notag\\&+ \gamma^2 \mc V(\state{t+2},\bar {\pi}) \vert \state{t} = \state{} ]\notag\\
	\leq & \mathbb E_{\bar {\pi}} [ \reward{}(\state{},\pi^{\dagger}(\state{};\omega)) + \gamma r(\VEC s_{t+1}, \pi ^{\dagger}) + \gamma^2 r(\VEC s_{t+2}, \pi ^{\dagger}) \notag \\
	& + \gamma ^3 \mc V(\state{t+3},\bar {\pi}) \vert \state{t} = \state{}] \notag\\
	 &\cdots \notag\\	
	\leq &  \mathbb E_{\bar {\pi}} [ \reward{}(\state{},\pi^{\dagger}(\state{};\omega)) + \gamma r(\VEC s_{t+1}, \pi ^{\dagger})+ \gamma^2 r(\VEC s_{t+2}, \pi ^{\dagger}) \notag \\
	& + \gamma ^3 r(\VEC s_{t+3}, \pi ^{\dagger}) + \cdots \vert \state{t} = \state{} ]\notag\\
	= &\mc V(\state{},\pi^{\dagger}),
\end{align}
which completes the proof.
\end{proof}
The proof can be extended to all states and possible actions for comparing policy of ACSA with the baseline policy. Since $\mc V (\state{},\policynew)$ represents the expected sum log-rate of IAB networks, a better rate performance can be guaranteed by ACSA compared with the baseline policy $\Bar{\pi}$ if Proposition 1 is satisfied.

\section{Simulation Results}\label{sec:result}
In this section, we present simulation results to demonstrate the performance of proposed methods for proportional fairness spectrum allocation in different \gls{iab} scenarios. 
\subsection{Simulation Setup}
We consider an \gls{iab} network consisting of one \gls{mbs} at the center, and $\Nsbs$ \gls{iab} nodes which follow the Poisson point process (PPP) deployed at the radius of $250$-meter from \gls{mbs}. \gls{ue} groups are located at the radius of $150$-meter from its associated \gls{bs}, the initial locations of which follow PPP. We adopt random walk model for simulating mobility of \gls{ue}s in the IAB networks. The moving speed of \gls{ue} $u_j$ at time $t$ follows uniform distribution, i.e., $\nu_j \!\sim \mathcal U(0,2)$ m/s, while the moving angel follows $\psi_j \!\sim \mathcal U(0,2\pi)$. The distance-based path loss is characterized by the \gls{los} model for urban environments at $2.4$ Ghz, and is compliant with the LTE standard \cite{access2011radio}. The simulator initializes an \gls{iab} network setup according to the parameters shown in Table.~\ref{table:sim}. We consider the proportional fairness schedule which aims to maximize $\sum_{j\in\mc U} \log(\rate{j})$. The required rate is unique among all UEs, where {\color{black}$\Omega_{j} = 5,\forall j\in \mc U$}.\\

We next {\color{black}describe} the hyper-parameters used in our \gls{ddqn} and ACSA methods. To better compare two algorithms, we apply the same network structure in \gls{ddqn} and in critic of ACSA, which composes of three-hidden layers of fully-connected neural networks with $500,1000,500$ neurons. This structure is also applied to the actor network in ACSA. We found that employing different number of neurons at each hidden layer renders similar rate performance for the \gls{iab} network with small size and few sub-channels. However, the stability significantly degrades as the action space increases when few neurons are adopted. In the ACSA, the third-hidden layer is set as a customize activation function to fulfill the constraint condition in problem (P1). The softmax function is applied on each column of allocation matrix $\mbsmtx$ for \gls{mbs} since the constraint of the sum probability for allocating one typical channel equals to one. Correspondingly, this is to select the best "candidate" for each channel to get a higher score in critic networks. The sigmoid function is applied on each row of $\sbsmtx$ for the \gls{iab} node to select the best "combination" of sub-channels for allocation. The discount rate is $\gamma = 0.9$ and the exploration rate is $\epsilon = 0.9$ with a decay rate of $0.9995$. The learning rate is set to be the same for both algorithms as $\beta = 10^{-4}$. The soft update rate in ACSA model is $\tau = 0.01$.

The input of the proposed models is the observed state from the environment where $\state{t}$. The output of \gls{ddqn} is a single action for spectrum allocation. This action decision is converted to the mapping matrix $\ls\mbsmtx, \sbsmtx\rs^T$ by a post-process. The output of the ACSA is the processed action, which can be directly applied to the allocation vectors. Thus, the output action $\actionv{}$ is with size
$(2 \Nsbs +1)\times \Nsb$, where the first $(\Nsbs  +1)\times \Nsb$ is the spectrum allocation $\mbsmtx$ for \gls{mbs}, and the next $\Nsbs\times \Nsb$ is the allocation $\sbsmtx$ for all \gls{iab} nodes. 



\begin{table}[t]
\centering
\caption{Simulator Setup}
\label{table:sim}
\scalebox{0.82}{\begin{tabular}[t]{c|l}
Parameter & Value \\
\hline
Carrier Frequency & $2.4$ Ghz \\
Bandwidth $W$ & $20$ Mhz \\
\gls{mbs} Pathloss & $34+40\log(d)$ \\
\gls{iab} Pathloss & $37+30\log(d)$ \\
Subchannel & Rayleigh Fading \\
Transmit power at \gls{mbs} & $43$ dBm \\
Transmit power at \gls{iab} & $33$ dBm \\
Self-interference & $-70$dB \\
Spectral Noise & $-174 \text{dBm/Hz} + 10\log(W) + 10\text{dB}$ \\
\end{tabular}}
\end{table}

\subsection{Learning to Improve Fariness Schedule}


\begin{figure} [b] 
	 \vskip 0.2in
	\begin{center}
		\centerline{\includegraphics[width=88mm]{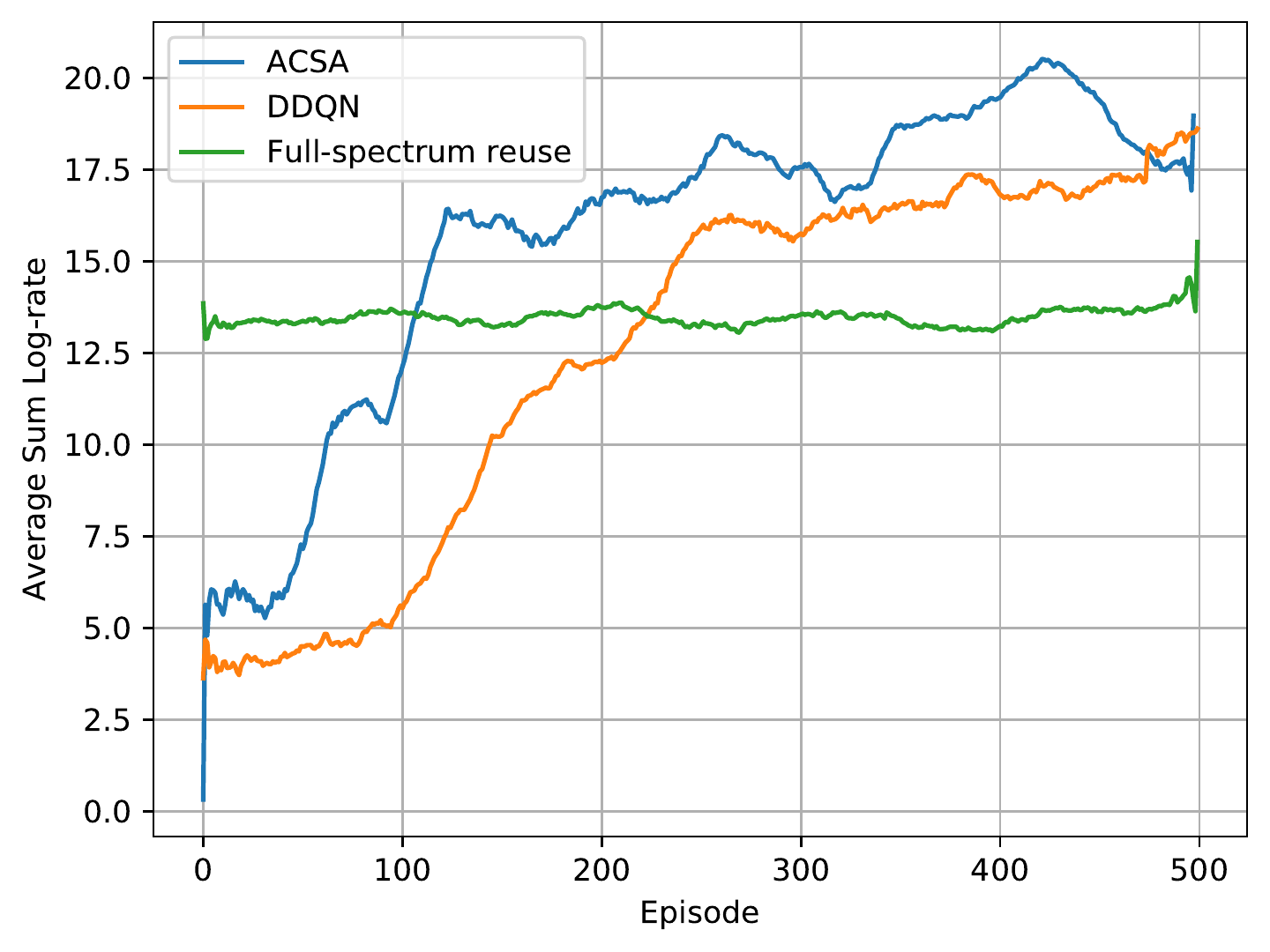}}
		\caption{Average sum log-rate for $L = 4, M = 10$}
		\label{fig:reward}
	\end{center}
	 \vskip -0.2in
\end{figure}

We present a typical result with $\Nsbs = 4$ \gls{iab} nodes, which are deployed around \gls{mbs}. The available sub-channels for resource allocation at each \gls{bs} are $\Nsb = 10$. As illustrated in Fig.~\ref{fig:reward}, we compare the sum log-rate for the proposed \gls{ddqn} and ACSA algorithms with the full-spectrum reuse strategy. The results show that both proposed methods can effectively learn from interacting with the time-varying environment and achieve better log sum-rate performance than the full-spectrum reuse strategy. It shows that reusing full-spectrum can not guarantee the optimal log sum-rate for \gls{iab} networks due to the severe interference introduced by turning on all sub-channels at \gls{iab} nodes and MBS. This interference not only affects co-tier \gls{iab} nodes during backhaul transmission, but also affects the \gls{ue} at \gls{mbs}. In addition, the convergence speed of ACSA exceeds that of \gls{ddqn}. This shows that the ACSA is more suitable for solving the problem (P1) with IAB networks. For the proposed methods, it can be seen that their performance keeps low in the beginning of learning phase. This is because the agent is exploring the environment with trial-and-error experiments. After around $120$ episodes for the ACSA and $200$ episodes for the \gls{ddqn}, the performance of both methods intend to increase and become relatively stable. However, due to the mobility of UEs, the sum log-rate of all methods fluctuates. Unlike the static full-reuse strategy, the \gls{drl}-based methods can be trained to dynamically allocate limited spectrum resource according to the current link conditions in the IAB network, which hence achieve a better sum log-rate performance \cite{munk2016learning}. 


\begin{figure} [t] 
	 \vskip 0.2in
	\begin{center}
		\centerline{\includegraphics[width=86mm]{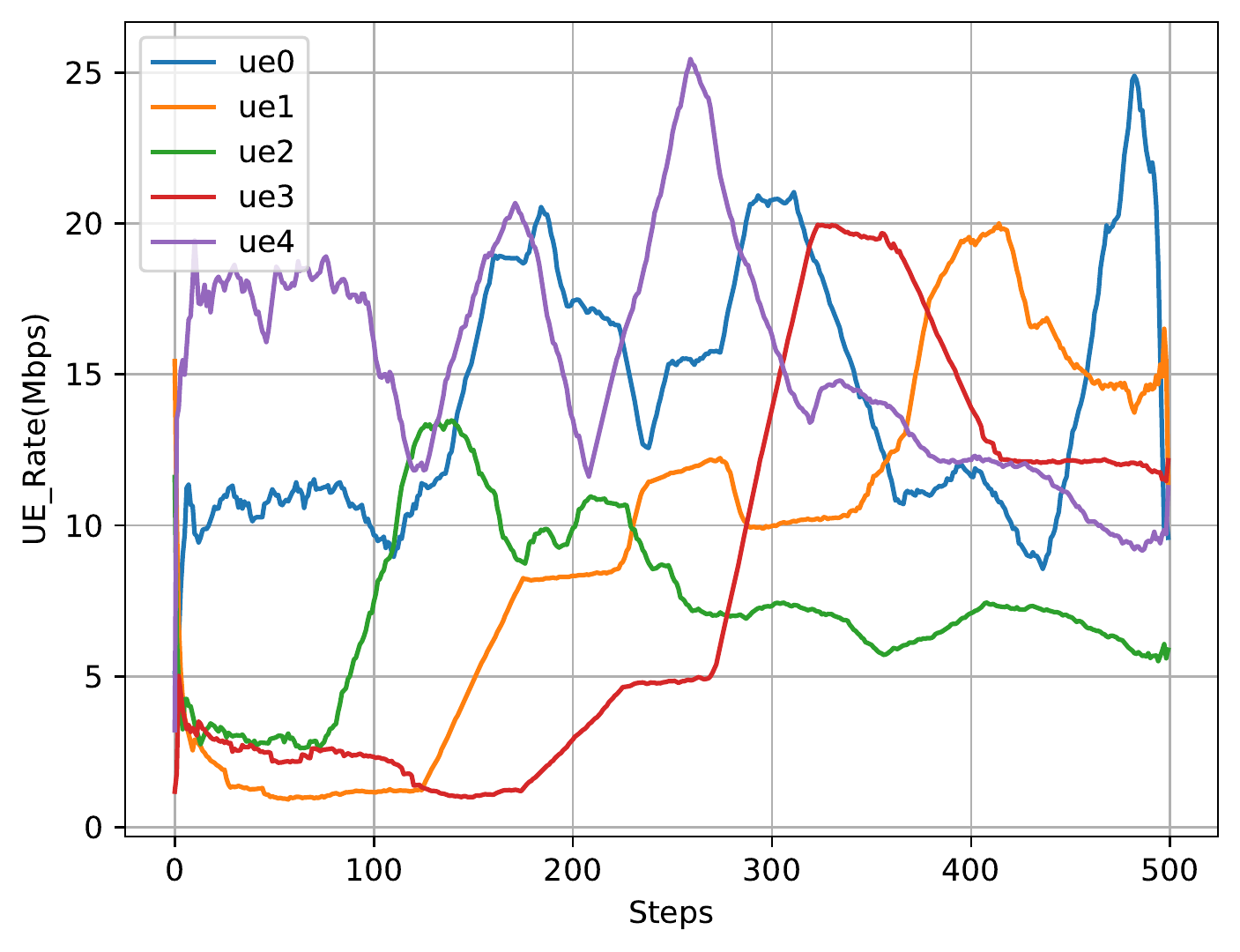}}
		\caption{Average rate of each \gls{ue} applying ACSA}
		\label{fig:uerate}
	\end{center}
	 \vskip -0.2in
\end{figure} 

Fig.~\ref{fig:uerate} illustrates one instance of the average rate of individual \gls{ue} using ACSA method. The log sum-rate is relatively low in the beginning since only two \gls{ue}s achieve rate requirements. This is possible since at some specific time periods, the \gls{ue}s that locate close to neighbouring \gls{bs}s have worse link condition than other \gls{ue}s. Then the agent learns to adapt the allocation strategy to improve the log-sum rate. It can be seen that the rates of UE1, UE2 and UE3 increase rapidly after $100$ steps. After around $300$ steps, all \gls{ue}s have achieved their system requirement. However, with increasing the rates of UE1 and UE3, the rate of UE2 is slightly decreased. This is due to that the system has reached an equilibrium allocation solution, where no \gls{ue} can improve their individual rate without decreasing the rate of at least one other \gls{ue}. 





\begin{figure} [t] 
	 \vskip 0.2 in
	\begin{center}
		\centerline{\includegraphics[width=82mm]{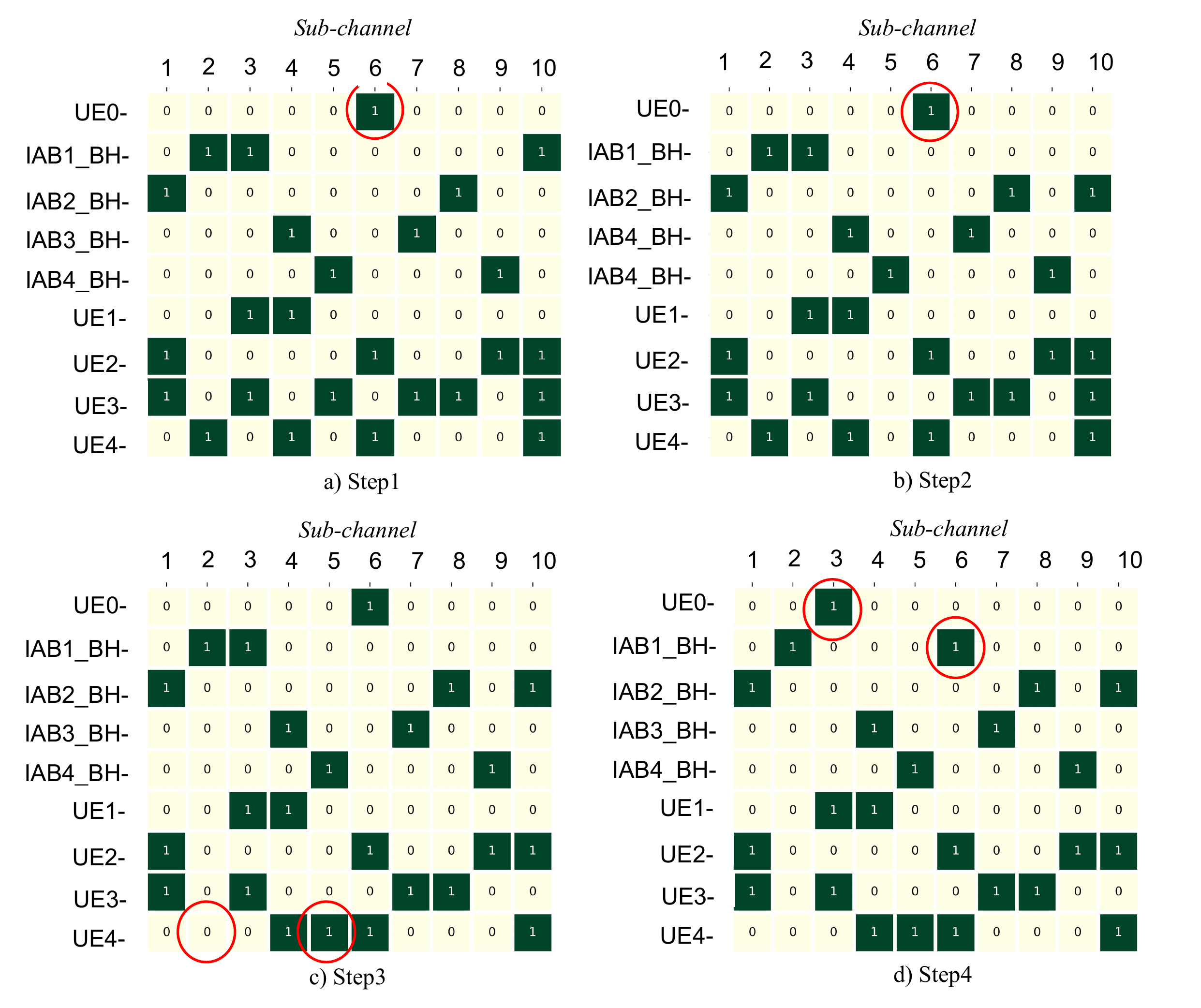}}
		\caption{Spectrum allocation in 4 snapshots}
		\label{fig:alloc}
	\end{center}
	 \vskip -0.2in
\end{figure}

\begin{figure} [t] 
	 \vskip 0.2in
	\begin{center}
		\centerline{\includegraphics[width=90mm]{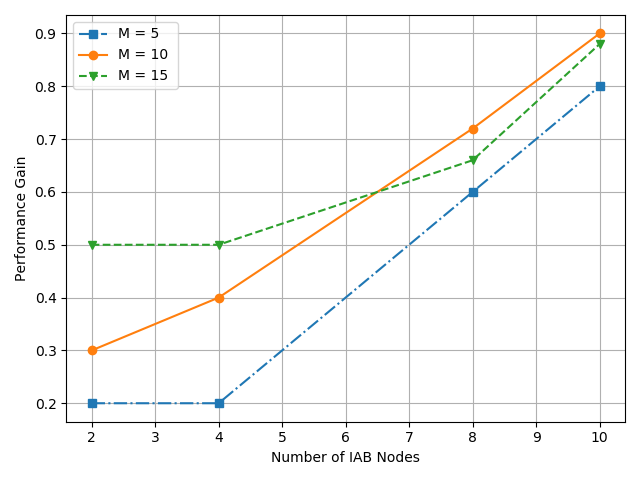}}
		\caption{{\color{black}Average performance gain using ACSA with different number of IAB nodes}}
		\label{fig:gain}
	\end{center}
	 \vskip -0.2in
\end{figure}

To better understand the learning process of \gls{drl} methods, we show how allocation decisions are made in consecutive $4$ time steps in one test case, {\color{black}which is presented} in Fig.~\ref{fig:alloc}. It can be seen that the $6$-th sub-channel is allocated for UE0 associated with MBS at the first 3 time steps. However, at the last time-step, the $3$-rd sub-channel is allocated to UE0 instead of sub-channel $6$. For all the 4 steps, UE0 experiences cross-tier interference from other two transmitting \gls{iab} nodes. {\color{black}As numerical result shown at this test case, as UE0 moves dynamically in the environment, the channel gain between UE0 and interfering IAB nodes $g_{i,u_0}, i\in \mc B^-$ varies. From step 3 to step 4, the interference introduces by IAB node1 and node3 decreases while by which IAB node2 and node4 increases. To maximize the expected cumulative sum log-rate in the system, the optimal resource mapping at step4 is to allocate sub-channel 3 to UE0 and allocate sub-channel 6 to IAB node1 in order to guarantee the access link rate for UE0 as well as backhaul link rate for IAB node1.}
 
Besides, we evaluate the {\color{black}{performance gain} of ACSA to a fixed policy by {\color{black}simulating the average sum log-rates over $500$ episodes with different number of \gls{iab} nodes.} The performance gain is defined as $Gain = (\sum_{j\in \mc U} f^{\text{ACSA}}(C_j)-\sum_{j\in \mc U} f^{\text{fixed}}(C_j))/\sum_{j\in \mc U} f^{\text{fixed}}(C_j)$}.  We choose 4 representative scenarios with $L = 2,4,8,10$. The results in Fig. 6 show that as we increase the number of \gls{iab} nodes, the proposed \gls{drl} method can outperform the fixed allocation strategy. Since with increasing \gls{iab} nodes, the performance degradation of the static strategy  dominates that of the proposed method, the performance gain increases with $L$. As more spectrum is available in the network, the proposed method shows to perform better in interference mitigation than the static allocation method. 

\section{Conclusions}
We investigate the spectrum allocation in \gls{iab} networks by characterizing the optimization problem to maximize the sum log-rate while guaranteeing \gls{ue}s system requirements. Specifically, we propose two \gls{drl}-based algorithms to solve the optimization problem at hand. We show that both algorithms can learn to adapt allocation strategy for backhaul and access transmission with dynamic network settings. Moreover, the proposed ACSA is {\color{black}{more efficient in convergence speed compared with \gls{ddqn}}}, and it achieves considerable performance gain in system sum log-rate comparing to static strategy especially when the size of IAB networks gets larger. Simulation results are given to show promising potentials and better sum log-rate performance of proposed methods compared with the traditional full-spectrum reuse strategy.

\bibliography{Reflibb}
\bibliographystyle{IEEEtran}
\end{document}

%% file: MyDef-IRCsplit-Zbased.tex
\usepackage{mathrsfs}

\makeatletter
\newcommand*{\rom}[1]{\expandafter\@slowromancap\romannumeral #1@}
\makeatother

\newcommand{\mb}[1]{\mathbf{#1}}
\newcommand{\mc}[1]{\mathcal{#1}}

\newcommand{\sinr}{\Gamma}

\newcommand{\lver}{\left\vert}
\newcommand{\rver}{\right\vert}

\newcommand{\lp}{\left(}
\newcommand{\rp}{\right)}
\newcommand{\ls}{\left[}
\newcommand{\rs}{\right]}

\newcommand{\card}[1]{\lver #1\rver}

\newcommand{\bmat}{\begin{bmatrix}}
\newcommand{\emat}{\end{bmatrix}}
\newcommand{\MAT}[1]{\mb{#1}}
\newcommand{\VEC}[1]{\mb{#1}}

\newcommand{\nextstate}{\VEC{s}'}
\newcommand{\nextaction}{\VEC{a}'}
\newcommand{\thetaeval}{\VEC{\theta}}
\newcommand{\thetatarget}{\VEC{\theta}^{-}}
\newcommand{\thetapolicy}{\VEC{\omega}}
\newcommand{\thetaq}{\VEC{\theta}}
\newcommand{\policybase}{\VEC{\bar{\pi}}}
\newcommand{\policynew}{\VEC{\pi}^\dagger}

\newcommand{\sbscvec}[1]{\VEC{z}_{#1}}
\newcommand{\mbsvvec}[1]{\VEC{x}_{#1}}
\newcommand{\state}[1]{\VEC{s}_{#1}}
\newcommand{\action}[1]{a_{#1}}
\newcommand{\actionv}[1]{\VEC{a}_{#1}}

\newcommand{\mbsmtx}{\MAT{X}}
\newcommand{\sbsmtx}{\MAT{Z}}

\newcommand{\Nsbs}{L} 
\newcommand{\Nsbsth}{l}
\newcommand{\Nsb}{M} 
\newcommand{\Nsbth}{m}

\newcommand{\Nsteps}{T}

\newcommand{\mbs}{b_0} 
\newcommand{\sbs}{b_l}
\newcommand{\mue}{u_0}
\newcommand{\sue}{u_l}
\newcommand{\pl}[2]{\alpha_{#1,#2}}
\newcommand{\chgain}[3]{g_{#1,#2}^{#3}} 
\newcommand{\ch}[3]{h_{#1,#2}^{#3}} 

\newcommand{\rate}[1]{C_{#1}}

\newcommand{\mbsv}[2]{x_{#1}^{#2}} 
\newcommand{\sbsc}[2]{z_{#1}^{#2}} 
\newcommand{\reward}[1]{r_{#1}}


%% file: MyAcronym-IRCsplit-Zbased.tex
\usepackage[acronym]{glossaries}

\newacronym{aoa}{AoA}{angles of arrival}
\newacronym{awgn}{AWGN}{additive white Gaussian noise}

\newacronym{bbu}{BBU}{baseband unit}
\newacronym{bs}{BS}{base station}

\newacronym{cdf}{CDF}{cumulative distribution function}
\newacronym{cee}{CEE}{channel estimation error}
\newacronym{cp}{CP}{cyclic prefix}
\newacronym{cpri}{CPRI}{common public radio interface}
\newacronym{cran}{C-RAN}{centralized radio access network}
\newacronym{csi}{CSI}{channel state information}

\newacronym{dft}{DFT}{discrete Fourier transform}
\newacronym{ds}{DS}{direction-selection}
\newacronym{du}{DU}{digital unit}
\newacronym{dl}{DL}{downlink}
\newacronym{drl}{DRL}{deep reinforcement learning}
\newacronym{dnn}{DNN}{deep neural network}
\newacronym{dqn}{DQN}{deep Q-network}
\newacronym{ddqn}{DDQN}{double deep Q-network}
\newacronym{ddpg}{DDPG}{deep deterministic policy gradient}

\newacronym{fft}{FFT}{fast Fourier transform}
\newacronym{fh}{FH}{Fronthaul}
\newacronym{phy}{PHY}{physical layer}
\newacronym{fd}{FD}{full-duplex}

\newacronym{hetnet}{HetNeT}{heterogeneous cellular networks}
\newacronym{irc}{IRC}{interference rejection combining}
\newacronym{ici}{ICI}{inter-cell interference}
\newacronym{iab}{IAB}{integrated access and backhaul}

\newacronym{lls}{LLS}{lower-layer split}
\newacronym{los}{LoS}{line-of-sight}
\newacronym{ls}{LS}{least square}
\newacronym{lte}{LTE}{Long Term Evolution}

\newacronym{mimo}{MIMO}{multiple-input-multiple-output}
\newacronym{mmse}{MMSE}{minimum mean-square error}
\newacronym{mrc}{MRC}{maximum ratio combining}
\newacronym{mumimo}{MU-MIMO}{multiuser-MIMO}
\newacronym{mbs}{DBS}{donor base station}
\newacronym{mt}{MT}{mobile termination}
\newacronym{mmwave}{mmWave}{millimeter wave}
\newacronym{mdp}{MDP}{Markov Decision Process}

\newacronym{nr}{NR}{New Radio}
\newacronym{nn}{NN}{Neural Network}

\newacronym{pca}{PCA}{Principal Component Analysis}

\newacronym{qos}{QoS}{quality of Service}

\newacronym{ran}{RAN}{radio access network}
\newacronym{rf}{RF}{radio frequency}
\newacronym{rru}{RRU}{remote radio unit}
\newacronym{rl}{RL}{reinforcement learning}
\newacronym{snr}{SNR}{signal-to-noise-power ratio}
\newacronym{sir}{SIR}{signal-to-interference-power ratio}
\newacronym{sinr}{SINR}{signal-to-interference-and-noise ratio}
\newacronym{svd}{SVD}{singular value decomposition}
\newacronym{si}{SI}{study item}

\newacronym{tco}{TCO}{total-cost-of-ownership}
\newacronym{td}{TD}{Temporal Difference}
\newacronym{ue}{UE}{user equipment}
\newacronym{ula}{ULA}{uniform linear array}
\newacronym{ul}{UL}{uplink}

\newacronym{zf}{ZF}{zero-forcing}

%% file: tccn_drlsa.bbl
\begin{thebibliography}{10}
\providecommand{\url}[1]{#1}
\csname url@samestyle\endcsname
\providecommand{\newblock}{\relax}
\providecommand{\bibinfo}[2]{#2}
\providecommand{\BIBentrySTDinterwordspacing}{\spaceskip=0pt\relax}
\providecommand{\BIBentryALTinterwordstretchfactor}{4}
\providecommand{\BIBentryALTinterwordspacing}{\spaceskip=\fontdimen2\font plus
\BIBentryALTinterwordstretchfactor\fontdimen3\font minus
  \fontdimen4\font\relax}
\providecommand{\BIBforeignlanguage}[2]{{%
\expandafter\ifx\csname l@#1\endcsname\relax
\typeout{** WARNING: IEEEtran.bst: No hyphenation pattern has been}%
\typeout{** loaded for the language `#1'. Using the pattern for}%
\typeout{** the default language instead.}%
\else
\language=\csname l@#1\endcsname
\fi
#2}}
\providecommand{\BIBdecl}{\relax}
\BIBdecl

\bibitem{DBLP:journals/corr/abs-1906-09298}
\BIBentryALTinterwordspacing
O.~M. Teyeb, A.~Muhammad, G.~Mildh, E.~Dahlman, F.~Barac, and B.~Makki,
  ``Integrated access backhauled networks,'' \emph{CoRR}, vol. abs/1906.09298,
  2019. [Online]. Available: \url{http://arxiv.org/abs/1906.09298}
\BIBentrySTDinterwordspacing

\bibitem{willebrand2001fiber}
H.~A. Willebrand and B.~S. Ghuman, ``Fiber optics without fiber,'' \emph{IEEE
  spectrum}, vol.~38, no.~8, pp. 40--45, 2001.

\bibitem{saha2018bandwidth}
C.~Saha, M.~Afshang, and H.~S. Dhillon, ``Bandwidth partitioning and downlink
  analysis in millimeter wave integrated access and backhaul for 5g,''
  \emph{IEEE Transactions on Wireless Communications}, vol.~17, no.~12, pp.
  8195--8210, 2018.

\bibitem{saha2018integrated}
------, ``Integrated mmwave access and backhaul in 5g: Bandwidth partitioning
  and downlink analysis,'' in \emph{2018 IEEE International Conference on
  Communications (ICC)}.\hskip 1em plus 0.5em minus 0.4em\relax IEEE, 2018, pp.
  1--6.

\bibitem{polese2018end}
M.~Polese, M.~Giordani, A.~Roy, S.~Goyal, D.~Castor, and M.~Zorzi, ``End-to-end
  simulation of integrated access and backhaul at mmwaves,'' in \emph{2018 IEEE
  23rd International Workshop on Computer Aided Modeling and Design of
  Communication Links and Networks (CAMAD)}.\hskip 1em plus 0.5em minus
  0.4em\relax IEEE, 2018, pp. 1--7.

\bibitem{polese2018distributed}
M.~Polese, M.~Giordani, A.~Roy, D.~Castor, and M.~Zorzi, ``Distributed path
  selection strategies for integrated access and backhaul at mmwaves,'' in
  \emph{2018 IEEE Global Communications Conference (GLOBECOM)}.\hskip 1em plus
  0.5em minus 0.4em\relax IEEE, 2018, pp. 1--7.

\bibitem{chandrasekhar2009spectrum}
V.~Chandrasekhar and J.~G. Andrews, ``Spectrum allocation in tiered cellular
  networks,'' \emph{IEEE Transactions on Communications}, vol.~57, no.~10, pp.
  3059--3068, 2009.

\bibitem{zhuang2018large}
B.~Zhuang, D.~Guo, E.~Wei, and M.~L. Honig, ``Large-scale spectrum allocation
  for cellular networks via sparse optimization,'' \emph{IEEE Transactions on
  Signal Processing}, vol.~66, no.~20, pp. 5470--5483, 2018.

\bibitem{zhou2020blockchain}
Z.~\color{black} Zhou, X.~Chen, Y.~Zhang, and S.~Mumtaz, ``Blockchain-empowered
  secure spectrum sharing for 5g heterogeneous networks,'' \emph{IEEE Network},
  vol.~34, no.~1, pp. 24--31, 2020\color{black}.

\bibitem{yang2020heterogeneous}
Q.~\color{black}Yang, T.~Jiang, N.~C. Beaulieu, J.~Wang, C.~Jiang, S.~Mumtaz,
  and Z.~Zhou, ``Heterogeneous semi-blind interference alignment in finite-snr
  networks with fairness consideration,'' \emph{IEEE Transactions on Wireless
  Communications}, 2020\color{black}.

\bibitem{zhang20196g}
L.~Zhang, Y.-C. Liang, and D.~Niyato, ``6g visions: Mobile ultra-broadband,
  super internet-of-things, and artificial intelligence,'' \emph{China
  Communications}, vol.~16, no.~8, pp. 1--14, 2019.

\bibitem{silver2016mastering}
D.~Silver, A.~Huang, C.~J. Maddison, A.~Guez, L.~Sifre, G.~Van Den~Driessche,
  J.~Schrittwieser, I.~Antonoglou, V.~Panneershelvam, M.~Lanctot \emph{et~al.},
  ``Mastering the game of go with deep neural networks and tree search,''
  \emph{nature}, vol. 529, no. 7587, p. 484, 2016.

\bibitem{watkins1992q}
C.~J. Watkins and P.~Dayan, ``Q-learning,'' \emph{Machine learning}, vol.~8,
  no. 3-4, pp. 279--292, 1992.

\bibitem{ghadimi2017reinforcement}
E.~Ghadimi, F.~D. Calabrese, G.~Peters, and P.~Soldati, ``A reinforcement
  learning approach to power control and rate adaptation in cellular
  networks,'' in \emph{2017 IEEE International Conference on Communications
  (ICC)}.\hskip 1em plus 0.5em minus 0.4em\relax IEEE, 2017, pp. 1--7.

\bibitem{zhang2018power}
Y.~Zhang, C.~Kang, T.~Ma, Y.~Teng, and D.~Guo, ``Power allocation in multi-cell
  networks using deep reinforcement learning,'' in \emph{2018 IEEE 88th
  Vehicular Technology Conference (VTC-Fall)}.\hskip 1em plus 0.5em minus
  0.4em\relax IEEE, 2018, pp. 1--6.

\bibitem{zhang2019deep}
L.~Zhang, J.~Tan, Y.-C. Liang, G.~Feng, and D.~Niyato, ``Deep reinforcement
  learning-based modulation and coding scheme selection in cognitive
  heterogeneous networks,'' \emph{IEEE Transactions on Wireless
  Communications}, vol.~18, no.~6, pp. 3281--3294, 2019.

\bibitem{lillicrap2015continuous}
T.~P. Lillicrap, J.~J. Hunt, A.~Pritzel, N.~Heess, T.~Erez, Y.~Tassa,
  D.~Silver, and D.~Wierstra, ``Continuous control with deep reinforcement
  learning,'' \emph{arXiv preprint arXiv:1509.02971}, 2015.

\bibitem{dulac2015deep}
G.~Dulac-Arnold, R.~Evans, H.~van Hasselt, P.~Sunehag, T.~Lillicrap, J.~Hunt,
  T.~Mann, T.~Weber, T.~Degris, and B.~Coppin, ``Deep reinforcement learning in
  large discrete action spaces,'' \emph{arXiv preprint arXiv:1512.07679}, 2015.

\bibitem{luo2008dynamic}
Z.-Q. Luo and S.~Zhang, ``Dynamic spectrum management: Complexity and
  duality,'' \emph{IEEE journal of selected topics in signal processing},
  vol.~2, no.~1, pp. 57--73, 2008.

\bibitem{srikant2013communication}
R.~Srikant and L.~Ying, \emph{Communication networks: an optimization, control,
  and stochastic networks perspective}.\hskip 1em plus 0.5em minus 0.4em\relax
  Cambridge University Press, 2013.

\bibitem{sutton1998introduction}
R.~S. Sutton, A.~G. Barto \emph{et~al.}, \emph{Introduction to reinforcement
  learning}.\hskip 1em plus 0.5em minus 0.4em\relax MIT press Cambridge, 1998,
  vol.~2, no.~4.

\bibitem{kaelbling1996reinforcement}
L.~P. Kaelbling, M.~L. Littman, and A.~W. Moore, ``Reinforcement learning: A
  survey,'' \emph{Journal of artificial intelligence research}, vol.~4, pp.
  237--285, 1996.

\bibitem{tsitsiklis1994asynchronous}
J.~N. Tsitsiklis, ``Asynchronous stochastic approximation and q-learning,''
  \emph{Machine learning}, vol.~16, no.~3, pp. 185--202, 1994.

\bibitem{DBLP:journals/corr/MnihKSGAWR13}
\BIBentryALTinterwordspacing
V.~Mnih, K.~Kavukcuoglu, D.~Silver, A.~Graves, I.~Antonoglou, D.~Wierstra, and
  M.~A. Riedmiller, ``Playing atari with deep reinforcement learning,''
  \emph{CoRR}, vol. abs/1312.5602, 2013. [Online]. Available:
  \url{http://arxiv.org/abs/1312.5602}
\BIBentrySTDinterwordspacing

\bibitem{DBLP:journals/corr/HasseltGS15}
\BIBentryALTinterwordspacing
H.~van Hasselt, A.~Guez, and D.~Silver, ``Deep reinforcement learning with
  double q-learning,'' \emph{CoRR}, vol. abs/1509.06461, 2015. [Online].
  Available: \url{http://arxiv.org/abs/1509.06461}
\BIBentrySTDinterwordspacing

\bibitem{schaul2015prioritized}
T.~Schaul, J.~Quan, I.~Antonoglou, and D.~Silver, ``Prioritized experience
  replay,'' \emph{arXiv preprint arXiv:1511.05952}, 2015.

\bibitem{wang2015dueling}
Z.~Wang, T.~Schaul, M.~Hessel, H.~Van~Hasselt, M.~Lanctot, and N.~De~Freitas,
  ``Dueling network architectures for deep reinforcement learning,''
  \emph{arXiv preprint arXiv:1511.06581}, 2015.

\bibitem{sutton2000policy}
R.~S. Sutton, D.~A. McAllester, S.~P. Singh, and Y.~Mansour, ``Policy gradient
  methods for reinforcement learning with function approximation,'' in
  \emph{Advances in neural information processing systems}, 2000, pp.
  1057--1063.

\bibitem{degris2012model}
T.~Degris, P.~M. Pilarski, and R.~S. Sutton, ``Model-free reinforcement
  learning with continuous action in practice,'' in \emph{2012 American Control
  Conference (ACC)}.\hskip 1em plus 0.5em minus 0.4em\relax IEEE, 2012, pp.
  2177--2182.

\bibitem{silver2014deterministic}
D.~Silver, G.~Lever, N.~Heess, T.~Degris, D.~Wierstra, and M.~Riedmiller,
  ``Deterministic policy gradient algorithms,'' 2014.

\bibitem{access2011radio}
E.~U. T.~R. Access, ``Radio frequency (rf) system scenarios,'' \emph{Release},
  vol.~8, p.~56, 2011.

\bibitem{munk2016learning}
J.~Munk, J.~Kober, and R.~Babu{\v{s}}ka, ``Learning state representation for
  deep actor-critic control,'' in \emph{2016 IEEE 55th Conference on Decision
  and Control (CDC)}.\hskip 1em plus 0.5em minus 0.4em\relax IEEE, 2016, pp.
  4667--4673.

\end{thebibliography}
